\documentclass[conference]{IEEEtran}

\makeatletter
\def\ps@headings{%
\def\@oddhead{\mbox{}\scriptsize\rightmark \hfil \thepage}%
\def\@evenhead{\scriptsize\thepage \hfil \leftmark\mbox{}}%
\def\@oddfoot{}%
\def\@evenfoot{}}
\makeatother

\usepackage{amssymb,amsmath}
\usepackage{paralist}
\usepackage{graphicx}
\usepackage{color}
\usepackage{cite}

\newcommand{\F}{\mathbf{F}}

\newcommand{\C}{\mathcal{C}}
\newcommand{\N}{\mathcal{N}}

\newtheorem{theorem}{\textbf{Theorem}}

\newtheorem{definition}[theorem]{\textbf{Definition}}

\newtheorem{corollary}[theorem]{Corollary}

\newcommand{\nix}[1]{}

\begin{document}
\title{Network Coding-Based Protection Strategy\\
Against Node Failures}
\author{
\authorblockN{Salah A. Aly}
\authorblockA{Department of ECE \\Iowa State University\\  Ames, IA 50011, USA\\Email:salah@iastate.edu}
\and
\authorblockN{Ahmed E. Kamal}
\authorblockA{Department of ECE \\Iowa State University\\Ames, IA 50011, USA\\Email:
kamal@iastate.edu}
 } \maketitle

\begin{abstract}
The enormous increase in the usage of communication networks has made protection against node and link failures  essential in the
deployment of reliable networks.
To prevent loss of data due to node
failures, a network protection strategy is proposed that aims to withstand
such failures.
Particularly, a protection strategy against any single node failure is designed for a given network with a set of $n$ disjoint paths between senders and receivers. Network coding and reduced capacity are deployed in this
strategy without adding extra working paths to the readily available
connection paths.
This strategy is based on protection against node failures as  protection against multiple link failures.
 In addition, the encoding and decoding operational  aspects of the
premeditated protection strategy are demonstrated.
\end{abstract}

\section{Introduction}\label{sec:intro}
With the increase in the capacity of backbone networks, the failure of
a single link or node can result in the loss of enormous amounts of
information, which may lead to catastrophes, or at least loss of
revenue.
Network connections are therefore provisioned such that
they can survive such failures.
Several techniques to provide network survivability have been
introduced in the literature. Such
techniques either add extra resources, or reserve some of the
available network resources as backup circuits, just for the sake of
recovery from failures.
Recovery from failures is also required to be agile in order to
minimize the network outage time.
This recovery usually involves two steps: fault diagnosis and
location, and rerouting connections.
Hence, the optimal network survivability problem is
a multi-objective problem in terms of resource efficiency, operation cost,
and agility~\cite{zeng07}.

 Network coding
allows the intermediate nodes not only to forward packets using network
scheduling algorithms, but also encode/decode them using algebraic
primitive operations, see~\cite{ahlswede00,fragouli06,soljanin07,yeung06}
and the references therein. As an application of network coding, data loss
because of failures in communication links can be detected and recovered
if the sources are allowed to perform network coding operations~\cite{sprintson07}.

In network survivability, the four different types of failures that might
affect network operations are\cite{somani06,zhou00}: \begin{inparaenum} \item  link failure,
\item node failure, \item shared risk link group (SRLG) failure, and
    \item network control system failure.
  \end{inparaenum}
Henceforth, one needs to design network protection strategies against
these types of failures.  Although the common frequent failures are link
failures, node failures sometimes happen due to burned swritch/router,
fire, or any other hardware damage. In addition, the failure might be due to
network maintenance. However, node failure is more damaging than
link or system failures since multiple connections may be affected
by the failure of a single node.

Recently, the authors have proposed employing the network coding
technique in order to protect against
single and multiple link failures
\cite{aly08j,aly08i,kamal07a}, in a manner that achieves both agility
and resource efficiency. The idea is to form linear combinations of
data packets transmitted on the working circuits, and transmit these
combinations simultaneously on a shared protection circuit.
The protection circuit can take the form of an additional p-cycle~\cite{kamal06a, kamal07a}, a path  or a general tree
network~\cite{kamal08b}.
In the case of failures, the linear combinations can be used by the end
nodes of the connection(s) affected by the failure(s) to recover the
lost data packets.
These network protection strategies against link failures using
network coding have been extended to use reduced capacities
instead of reserving, or even adding separate protection
circuits~\cite{aly08j,aly08i}.
The advantages of using network coding-based protection are twofold:
first, one set of protection circuits is shared between a number of
connections, hence leading to reduced protection cost; and second,
copies of data packets are transmitted on the shared protection
circuit after being linearly combined, hence leading to fast recovery
of lost data since failure detection and data rerouting are not
needed.

In this paper, we consider the problem of providing protection against node
failures by the means of  network coding and  the reduced capacity
techniques.
As a byproduct of this protection strategy, protection against any
single link failure is also guaranteed.
This is based on representing the node failure by the failure of
multiple links.
However, the failed links are not any arbitrary links.
Since working paths used by the connections that are protected
together are link disjoint, the links that need to be protected are
used by different connections.

\section{Network Model}\label{sec:model}
The following points highlight the network model and  main considerations.

\begin{figure}[t]
\begin{center}
\includegraphics[scale=1.3]{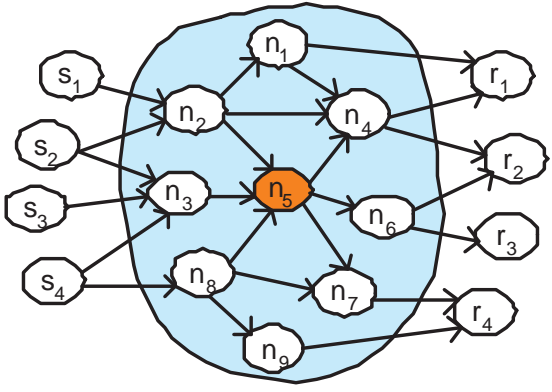}
\caption{
A Network $\N$ with  a set of nodes $\textbf{V}$ and a set of edges $E$.
The nodes $\textbf{V}$ consist of sources $S$, receivers $R$, and relay
nodes $V$. The node $n_5$ represents a failed node with 3 working
connections that must be protected at the failure
incidence.}\label{fig:relaynet}
\end{center}
\end{figure}

\begin{itemize}
\item Let $\N$ be a network represented by an abstract graph
    $G=(\textbf{V},E)$, where $\textbf{V}$ is the set of nodes and $E$
    be   set of undirected edges. Let $S$ and $R$ be  sets of independent sources
    and destinations, respectively. The set $\textbf{V}=V\cup S \cup
    R$ contains the relay nodes,
sources, and destinations, respectively,  as show in
    Fig.~\ref{fig:relaynet}. Assume for simplicity that $|S|=|R|=n$,
    hence the set of sources is equal to the set of receivers.

\item A path (connection) is a set of edges connected together with a
    starting node (sender) and an ending node (receiver).  \begin{eqnarray} L_i=
    \{(s_i,w_{1i}),(w_{1i},w_{2i}),\ldots,(w_{(m)i},r_i) \},\end{eqnarray} where
    $1\leq i\leq n$, $(w_{(j-1)i},w_{ji}) \in E$, and +ve integer m.
\item The node can be a router, switch, or an end terminal depending
    on the network model $\N$ and the transmission layer, see Fig.~\ref{fig:npaths}.

\item L is a set of paths $L=\{L_1,L_2,\ldots,L_n\}$ carrying the data
    from the sources to  receivers. Connection paths are link
    disjoint and provisioned in the network between  senders and
    receivers. All connections have the same bandwidth, otherwise a
    connection with a high bandwidth can be divided into multiple
    connections, each of which has the unit capacity. There are
    exactly $n$ connections.  A sender
    with a high capacity can divide its capacity into multiple unit
    capacities.

    \item We consider the case that the failures happen in the relay
        nodes.
The failures in the relay nodes might happen due to a failed
    switch, router, or any connecting point as shown in fig.~\ref{fig:relaynet}. We assume that the failures are independent of each other.
\end{itemize}

\begin{definition}[Node Relay Degree]\label{def:nodedegree}
Let $u$ be an arbitrary node in $V=\textbf{V}\backslash \{S \cup R\}$, which
relays the traffic between source and terminal nodes. The
number of connections passing through this node is called the
\textit{node relay degree}\/, and is referred to as $d (u)$. Put differently:
\begin{eqnarray}
d(u)=\big|\big\{ L_i: (u,w) \in L_i, \forall w \in \textbf{V}, 1\leq i\leq n\big\}\big|.
\end{eqnarray}

\end{definition}

Note that the above definition is different from the graph theoretic
definition of the node degrees; input and output degrees.
However, the node degree must not be less than the node relay degree. Furthermore, the node relay degree of a node $u$ is $d(u) \leq \lfloor \mu(u)/2\rfloor $, where $\mu(u)$ is the degree of a node u in an undirected graph.

We can define the network capacity from the min-cut max-flow
information theoretic view~\cite{ahlswede00}. It can be described as
follows.
\begin{definition}\label{def:capacitylink}
The unit capacity of a connecting path $L_i$ between $s_i$ and $r_i$ is defined
by \begin{eqnarray} c_i=\left\{
      \begin{array}{ll}
        1, & \hbox{$L_i$ is \emph{active};} \\
        0, & \hbox{otherwise.}
      \end{array}
    \right.
\end{eqnarray}
The total capacity of $\N$ is given by the summation of all path
capacities. What we mean by an \emph{active}  path is that the receiver is
able to receive and process packets throughout this path, see for further details~\cite{aly08preprint1}.
\end{definition}

 	\smallbreak

Clearly, if all paths are active then the total capacity of all
connections is $n$ and
the normalized capacity is $1$. If we assume there are n disjoint paths, then,
in general, the normalized capacity of the network for the active and failed paths is
computed by
\begin{eqnarray}
C_\N=\frac{1}{n}\sum_{i=1}^n c_i.
\end{eqnarray}

The\emph{ working paths} on a network with $n$ connection paths carry
traffic under normal operations, see Fig.~\ref{fig:npaths}. The \emph{Protection paths} provide an
alternate backup path to carry the traffic in case of failures. A
protection scheme ensures that data sent from the sources will reach the
receivers in case of failure incidences on the working paths~\cite{aly08i,aly08j}.

\section{Protection Against A Single Node Failure}\label{sec:SNF}

In this section we demonstrate a model for network protection against a
single node failure (SNF) using network coding. Previous work focused on
network protection against single and multiple link failures using
rerouting and sending packets throughout different
links. We use network coding and
reduced capacity on the paths carrying data from the sources to
destinations. The idea has been developed for the purpose of link and
path failures in~\cite{aly08i,kamal06a}. We present a
protection strategy  denoted by NPS-t.
Under NPS-t,
the normalized network capacity is based on the max-flow between sources
and destinations, and its given by $(n-t)/n$, where $t$ is the maximum
number of connections traversing any node in the network, i.e. in other words, it is the max node degree.
We develop the design methodology  of this
strategy.
In addition we derive bounds on the field size and
encoding operations.

Assume we have the same definitions as shown in the previous section. Let
$d(u)$ be the relay node degree of a node $u$ in
$V$. We define $d_0$ to be the $\max$ over all node's
relay degrees in the network $\N$.
\begin{eqnarray}
d_0= \max_{u \in V}  d (u)
\end{eqnarray}
Note that $d_o$ is the  degree  representing the max links that can
fail, in other words it is the number of working paths that might fail
due to the failure of a relay node.
Let $v$ be the node with relay degree $d_0$, and assume $v$ to be the
failed node. Our goal is
to protect the network $\N$ against this node failure.
In fact $d_0$ represents a set of failed connections caused by a failure of the
node $v$ in the network $\N$.
Although the failure of $v$ is represented by the failure of $2 d_0$ links,
each incoming link at $v$ has a corresponding outgoing link, and if
either, or both of these two links fail, the effect on the connection
is the same.
Therefore, our protection strategy is based on representing the node
failure by the failure of $d_0$ connections, and we therefore need
to protect against $d_0$ failed connections.
\begin{figure}[t]
 \begin{center} 
  \includegraphics[height=4cm,width=7cm]{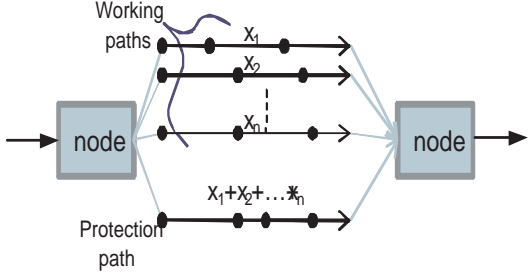}
  \caption{Network protection against a single path failure using reduced capacity and network coding.
   One path out of  $n$ primary paths  carries encoded data. The black points represent various other relay nodes}\label{fig:npaths}
\end{center}
\end{figure}

\subsection{NPS-t Protecting SNF with $d_0=t$ and Achieving $(n-t)/n$ Normalized Capacity}

Assume the sender $s_i$ sends
a message to the receiver $r_i$ via the path $L_i$. Also, assume without loss of
generality that $t$ disjoint working paths have failed due to the
failure of a single node.
Then, we describe protection code  as shown in
Scheme (\ref{eq:tfailures1}).
Under this protection scheme, $n-t$ of the working paths will
carry plain data units denoted by $x^i_j$'s, i.e. the
data unit transmitted on working path $j$ in round $i$.
The remaining $t$ paths will carry linear combinations, which are
denoted by $y_i$'s. They will be used to recover from data unit losses due to the failure.
\begin{figure}[h]
\begin{eqnarray}\label{eq:tfailures1}
\begin{array}{|c|cccccc|}
\hline
&1&2&3&4&\ldots&\lfloor n/t\rfloor\\
\hline    \hline
s_1 \rightarrow r_1 & y_1&x_1^1 &x_1^2&x_1^3&\ldots &x_{1}^{n-1} \\
s_2 \rightarrow r_2 &  y_2&  x_2^1 &x_2^2&x_2^3&\ldots&x_2^{n-1}  \\
\vdots&\vdots&\vdots&\vdots&\vdots&\vdots&\vdots\\
s_{t} \rightarrow r_{t} &  y_{t}&  x_{t}^1 &x_{t}^2&x_{t}^3&\ldots&x_{t}^{n-1}  \\
\!\! s_{t+1} \! \rightarrow \!\! r_{t+1} \!\! &\!\! x_{t+1}^1\!&\! y_{t+1}& \! x_{t+1}^2&\! x_{2t+1}^3&\! \ldots&\! x_{2t+1}^{n-1} \!\! \\
\vdots&\vdots&\vdots&\vdots&\vdots&\vdots&\vdots\\
s_{2t} \rightarrow r_{2t}&x_{2t}^1&y_{2t}& x_{2t}^2& x_{2t}^3&\ldots&x_{2t}^{n-1}  \\
\!\! s_{2t+1}\!\! \rightarrow \!\!r_{2t+1}\!\!&\!\!x_{2t+1}^1\!&\!x_{2t+1}^2&\! y_{2t+1}&\! x_{2t+1}^3& \ldots&\! x_{2t+1}^{n-1} \!\! \\
\vdots&\vdots&\vdots&\vdots&\vdots&\vdots&\vdots\\
\vdots&\vdots&\vdots&\vdots&\vdots&\vdots&\vdots\\
s_n \rightarrow r_n & x_n^1&x_n^2&x_n^3&x_n^4&\ldots&y_{n}\\
\hline
\end{array}
\end{eqnarray}
\end{figure}

In general, $y_\ell$ is given by
\begin{eqnarray}\label{eq:y_NPS-T}
y_\ell=\sum_{i=1}^{(j-1)t} a_i^\ell x_i^j + \sum_{i=jt+1}^n a_i^\ell
x_i^j  \nonumber \\ \mbox{   for  } (j-1)t+1 \leq \ell \leq jt,  1
\leq j \leq \lfloor \frac{n}{t}\rfloor.
\end{eqnarray}

We consider that the coefficients $a_i^l$'s are taken from
a finite field with $q >2$ alphabets.
Later in this section, we will show how to perform the encoding
and decoding operations for the purpose of recovery from failures. In addition, we will derive bounds on the field size in the next Section.
The following Theorem gives the normalized capacity of NPS-t strategy.
\begin{theorem}
Let $n$ be the total number of disjoint connections from sources to receivers. The
capacity of NPS-t strategy against $t$ path failures as a result of a
single node failure is given by
\begin{eqnarray}
\C_{\N}=(n-t)/(n)
\end{eqnarray}
\end{theorem}
\begin{proof}
In NPS-t, there are t paths that will carry encoded data in each round
time in a particular session. Without loss of generality, consider the case
in which $n/t$ is an integer\footnote{The general case in which $n/t$
  is not an integer can be accommodated by running the strategy for $m
  \lfloor (n/t) \rfloor$ rounds, where $m$ is the smallest integer
  such that $mn \mod t = 0$.} or assume that $\lfloor n/t \rfloor$. Therefore, there exists $(n/t)$ rounds,
in which the capacity
is $(n-t)$ in each round. Also, the capacity in the first round is $n-t$. Hence, we have
\begin{eqnarray}
\C_{\N} &=& \frac{\sum_{i=1}^{ n/t } (n-t)}{  (n/t)  n} \nonumber \\
&=& \frac{n-t}{n}
\end{eqnarray}
\end{proof}

The advantages of NPS-t strategy described in Scheme (\ref{eq:tfailures1}) are that:
\begin{itemize}
\item

The data is encoded and decoded online, and it will be sent and
received in different rounds. Once the receivers detect failures, they
are able to obtain a copy of the lost data  without delay
by querying the neighboring nodes with unbroken working paths.

\item The approach is suitable for applications that do not tolerate
    packet delay such as real-time applications, e.g., multimedia  and TV
    transmissions.
\item $\%100$ recovery against any single node failure is guaranteed.
In addition, up to $t$ disjoint path failures can be recovered from.

\item Using this strategy, no extra paths are needed. This will make
    this approach more suitable for applications, in which adding
    extra paths, or reserving links or paths just for protection, may
    not be feasible.

\item The encoding and decoding operations are linear, and the
    coefficients of the variables $x_i^j$'s are taken from a finite
    field with $q> 2$ elements.
\end{itemize}

\subsection{Encoding Operations}
Assume that each connection path
$L_i$ (L) has a unit capacity from a source  $s_i$ (S) to a
receiver $r_i$ (R). The data sent from the sources S to the receivers R
is transmitted in rounds.
Under NPS-t, in every round $n-t$ paths are used to carry plain data
($x_i^j$), and $t$ paths are used to carry protected data units.
there are $t$ protection paths.
Therefore, to treat all connections fairly,
there will be $\lfloor n/t \rfloor$ rounds, and in each round the
capacity is given by n-t.

We consider the case in which all symbols $x_i^j$'s  belong to the same round.
The first t sources transmit the first encoded data units
$y_1,y_2,\ldots,y_{t}$, and
in the second round, the next $t$ sources
transmit $y_{t+1},y_{t+2},\ldots,y_{2t}$, and so
on. All sources S and receiver R must keep
track of the round numbers. Let $ID_{s_i}$ and $x_{s_i}$ be the ID and
data initiated by the source $s_i$. Assume the round time $j$ in session
$\delta$ is given by $t^{j}_{\delta}$. Then the source $s_i$ will send
$packet_{s_i}$ on the working path $L_i$ which includes
\begin{eqnarray}
Packet_{s_i}=(ID_{s_i}, x_{i}^\ell, t^\ell_\delta)
\end{eqnarray}

Also, the source $s_j$, that transmits on the protection path, will
send a packet $packet_{s_j}$:
\begin{eqnarray}
Packet_{s_j}=(ID_{s_j}, y_k, t^\ell_\delta),
\end{eqnarray}
where $y_k$ is defined in~(\ref{eq:y_NPS-T}). Hence the protection
paths are used to protect the data transmitted in round $\ell$, which
are included in the $x^l_i$ data units.
The encoded data
$y_k$ is computed in a simple way where source $s_j$, for example,
will collect all
sources' data units, and using proper coefficients, will compute the
$y_k$ data units defined in Scheme~(\ref{eq:y_NPS-T}).
In this case every data unit
$x_i^{\ell}$ is multiplied by a unique coefficient $a_i \in
\F_q$. This will differentiate the
encoded data $y_i's$.
 So, we have a system of $t$ independent equations at each round time
 that will be  used to recover at most $t$ unknown variables.

\subsection{Proper Coefficients Selection}
One way to select the coefficients $a_j^\ell$'s in each round such that we
have a system of $t$ linearly independent equations is by using the matrix $H$ shown in~(\ref{bch:parity}). Let $q$ be the order of a finite field, and
$\alpha$ be the root of unity in $\F_q$. Then we can use this matrix   to define the coefficients of the senders as:
\begin{eqnarray}\label{bch:parity} H =\left[
\begin{array}{ccccc}1 &\alpha &\alpha^2 &\cdots &\alpha^{n-1}\\1
&\alpha^2 &\alpha^4 &\cdots &\alpha^{2(n-1)}\\\vdots& \vdots &\vdots
&\ddots &\vdots\\1 &\alpha^{t-1} &\alpha^{2(t-1)} &\cdots
&\alpha^{(t-1)(n-1)}\end{array}\right].\end{eqnarray}


We have the following assumptions about the encoding operations as shown in Scheme~(\ref{scheme:encoding}).
\begin{compactenum}
\item  Clearly if we have one failure $t=1$, then all coefficients will be
    one. The first sender will always choose the unit value in the first row.

\item  If we assume $d_0=t$, then the $y_1,y_2,\ldots,y_t$ equations in the first round   are written as:
\begin{eqnarray}
y_1=\sum_{i=t+1}^nx_i^1, \mbox{  }
y_2=\sum_{i=t+1}^n \alpha^{(i-1) }x_i^1,
\end{eqnarray}
%
%
\begin{eqnarray}\label{eq:tcofficients}
y_j=\sum_{i=t+1}^n \alpha^{i(j-1) \mod (q-1)}x_i^1,~~~ 1\leq j \leq t
\end{eqnarray}

\end{compactenum}
Therefore, the scheme that describes the encoding operations in the first
round for $t$ link failures can be described as

\begin{eqnarray}\label{scheme:encoding}
\begin{array}{|c|ccccc|}
\hline
& \multicolumn{5}{|c|}{\mbox{ round one, t failures }}\\ \hline
&y_1&y_2&y_3&\ldots&y_t\\
\hline    \hline
  s_1 \rightarrow r_1 & 1 &1&1&\ldots   &1\\
    s_2 \rightarrow r_2 &  1& \alpha&\alpha^2&\ldots&\alpha^{t-1} \\
s_3 \rightarrow r_3 & 1& \alpha^2& \alpha^4&\ldots&\alpha^{2(t-1)}\\
 \vdots&\vdots&\vdots&\vdots&\vdots&\vdots\\
    s_i \rightarrow r_i& 1 &\alpha^{i-1} &\ldots &\ldots& \alpha^{(i-1)(t-1)}\\
 \vdots&\vdots&\vdots&\vdots&\vdots&\vdots\\
   s_n \rightarrow r_n & 1&\alpha^{n-1}&\alpha^{2(n-1)}&\ldots&\alpha^{(t-1)(n-1)}\\
\hline
\end{array}
\end{eqnarray}
This Scheme gives the general theme to choose the coefficients at any particular round in any session. However, the encoded data $y_i$'s are defined as shown in Equation~(\ref{eq:tcofficients}). In other words, for the first round in session one, the coefficients of the plain data $x_1,x_2,\ldots,x_t$ are set to zero. The scheme can be extended directly to any encoded data $y_k$.


\subsection{Decoding Operations} We know that the coefficients
$a_1^\ell,a_2^\ell,\ldots,a_n^\ell$ are elements of a finite field, $\F_q$, hence
the inverses of these elements exist and they are unique.  Once
a node fails which causes $t$ data units to be lost, and once the
receivers receive $t$ linearly independent equations, they can
linearly solve these equations to obtain the $t$ unknown  data units.
At one
particular session j, we have three cases for the failures:
\begin{compactenum}[i)]
\item All $t$ link failures happened in the working paths, i.e. the working
    paths have failed to convey the messages $x_i^\ell$ in round $\ell$.
In this case, $n-t$ equations will be received, $t$ of which are linear
combinations of $n-t$ data units, and the remaining $n-2t$ are explicit
$x_i$ data units, for a total of $n-t$ equations in $n-t$ data units.
    In this case any $t$
    equations (packets) of the $t$ encoded packets can be used to recover
    the lost data.

    \item All $t$ link failures happened in the protection paths at the failed node. In this case, the
exact remaining $n-t$ packets are working paths and they do not experience
any failures. Therefore, no recovery operations are needed.
\item The third case is that the failure might happen in some working and
    protection paths simultaneously in one particular round in a session.
    The recover can be done using any $t$ protection paths as shown in case
    i.
\end{compactenum}

\section{Bounds on the Finite Field Size, $\F_q$}\label{sec:bounds}
In this section we derive lower and upper bounds on the alphabet size
required  for the encoding and decoding operations.  In the proposed
schemes we assume that unidirectional connections exist between the senders and
receivers, which the information can be exchanged with little cost.
The first result shows that the alphabet size required must be greater
than the number of connections that carry unencoded data.

\begin{theorem}
Let $n$ be the number of disjoint connections in the network model $\N$.
Then the
receivers are able to decode the encoded messages over $\F_q$ and will
recover from $t\geq 2$ path failures passing through if
\begin{eqnarray}
q\geq n-t+1.
\end{eqnarray}
Also, if $q=p^r$, then $r \leq \lceil \log_p(n+1) \rceil$. The binary field is sufficient in case of a single path failure.
\end{theorem}
\begin{proof}
We will prove the lower bound by construction. Assume a NPS-t at one
particular time $t_\delta^\ell$ in the round  $\ell$ in a certain session
$\delta$.  The protection code of NPS-t against t path failures is given
as
\begin{eqnarray}\label{eq:protectioncode1} C_{t} =\left[
\begin{array}{ccccc}1 &1 &1 &\cdots &1\\1 &\alpha &\alpha^2 &\cdots &\alpha^{n-1}\\1
&\alpha^2 &\alpha^4 &\cdots &\alpha^{2(n-1)}\\\vdots& \vdots &\vdots
&\ddots &\vdots\\1 &\alpha^{t-1} &\alpha^{2(t-1)} &\cdots
&\alpha^{(t-1)(n-1)}\end{array}\right]\end{eqnarray}

Without loss of generality, the interpretation of
Equation~(\ref{eq:protectioncode1}) is as follows:
\begin{compactenum}[i)]
\item The columns correspond to the senders $S$ and rows correspond to t
    encoded data $y_1,y_2,\ldots,y_t$.
\item The first row corresponds to $y_1$ if we assume the first
    round
    in
    session one. Furthermore, every row represents the coefficients of
    every senders at a particular round.
\item The column $i$ represents the coefficients of the sender $s_i$
    through all protection paths $L_1,L_2,\ldots,L_t$.
\item Any element $\alpha^i \in F_q$ appears once in a column and row,
    except in the first column and first row, where all elements are one's. All columns (rows) are linearly independent.

\end{compactenum}

Due to the fact that the $t$ failures might occur at any $t$ working paths of
$L=\{l_1,L_2,\ldots,l_n\}$, then we can not predict the $t$ protection paths
as well. This means that $t$ out of the $n$ columns do not participate in the
encoding coefficients, because $t$ paths will carry encoded data. We notice
that removing any $t$ out of the $n$ columns in
Equation~(\ref{eq:protectioncode1}) will result in $n-t$ different coefficients in each  column. Furthermore any $t$ columns will give a $\mu \times \mu$ square sub-matrix that has a full-rank, this will be proved in our extended work.  Therefore the smallest
    finite field that  satisfies this condition must have $n-t+1$
    elements.

  The upper bound comes from the case of no failures, hence $q \geq
    (n+1)$. Assume q is a prime, then the result follows.
\end{proof}

if $q=2^r$, then in general the previous bound can be stated as
\begin{eqnarray}
n-t+1 \leq q \leq 2^{\lceil \log_2(n+1) \rceil}.
\end{eqnarray}

We defined the feasible solution for the encoding and decoding operations of NPS-t as the solution that has integer reachable upper bounds.

\medskip

\begin{corollary}
The protection code~(\ref{eq:protectioncode1}) always gives a feasible solution.
\end{corollary}

\bigskip

\section{Conclusions}\label{sec:conclusion}
Protection against node and link failures are extensional in all communication networks.
In this paper, we presented a model for network protection against a single
node failure, which is equivalent to
protection against $t$ link failures,
and can therefore be used  to protect against $t$ link failures.
We demonstrated an implementation strategy
for the proposed network protection scheme.
The network capacity is estimated, and
bounds on the network resources are established.
Our future work will include approaches for
deploying the proposed protection strategy.

\bigskip

\section*{Acknowledgment}
This research was supported in part by grants CNS-0626741 and
CNS-0721453 from the National Science Foundation,
and a gift from Cisco Systems.

\scriptsize
\bibliographystyle{plain}

\begin{thebibliography}{10}

\bibitem{ahlswede00}
R.~Ahlswede, N.~Cai, S.-Y.~R. Li, and R.~W. Yeung.
\newblock Network information flow.
\newblock {\em IEEE Trans. Inform. Theory}, 46:1204--1216, 2000.

\bibitem{aly08i}
S.~A. Aly and A.~E. Kamal.
\newblock Network protection codes against link failures using network coding.
\newblock In {\em Proc. IEEE GlobelComm '08, New Orleans, LA}, December 1-4
  2008.
\newblock arXiv:0809.1258v1 [cs.IT].

\bibitem{aly08preprint1}
S.~A. Aly and A.~E. Kamal.
\newblock Network protection codes: Providing self-healing in autonomic
  networks using network coding.
\newblock {\em IEEE Journal on Selected Areas in Communciation (JSAC), special
  issue on Autonomic Communications, submitted}, 2008.
\newblock arXiv:0812.0972v1 [cs.NI].

\bibitem{aly08j}
S.~A. Aly and A.~E. Kamal.
\newblock Network coding-based protection strategies against a single link
  failure in optical networks.
\newblock In {\em Proc. IEEE International Conference on Computer Engineering
  \& Systems (ICCES'08), Cairo, EG}, November 25-27, 2008.
\newblock arXiv:0810.4059v1 [cs.IT].

\bibitem{fragouli06}
C.~Fragouli, J.~Le Boudec, and J.~Widmer.
\newblock Network coding: An instant primer.
\newblock {\em ACM SIGCOMM Computer Communication Review}, 36(1):63--68, 2006.

\bibitem{soljanin07}
C.~Fragouli and E.~Soljanin.
\newblock Network coding applications, foundations and trends in networking.
\newblock {\em Hanover, MA, Publishers Inc.}, vol. 2, no. 2, pp. 135-269, 2007.

\bibitem{kamal06a}
A.~E. Kamal.
\newblock {1+N} protection in optical mesh networks using network coding on
  p-cycles.
\newblock In {\em Proc. of the IEEE Globecom}, 2006.

\bibitem{kamal07a}
A.~E. Kamal.
\newblock {1+N} protection against multiple faults in mesh networks.
\newblock In {\em Proc. of the IEEE International Conference on Communications
  {(ICC)}}, 2007.

\bibitem{kamal08b}
A.~E. Kamal.
\newblock A generalized strategy for {1+N} protection.
\newblock In {\em Proc. of the IEEE International Conference on Communications
  (ICC)}, 2008.

\bibitem{somani06}
A.~K. Somani.
\newblock {\em Survivability and traffic grooming in Optical Networks}.
\newblock Cambridge Press, 2006.

\bibitem{sprintson07}
A.~Sprintson, S.~Y. El-Rouayheb, and C.~N. Georghiades.
\newblock Robust network coding for bidirected networks.
\newblock In {\em Proc. IEEE Information theory workshop, UCSD, CA}, pages
  378--383, 2007.

\bibitem{yeung06}
R.~W. Yeung, S.-Y.~R. Li, N.~Cai, and Z.~Zhang.
\newblock {\em Network Coding Theory}.
\newblock Now Publishers Inc., 2006.

\bibitem{zeng07}
H.~Zeng and A.~Vukovic.
\newblock The variant cycle-cover problem in fault detection and localization
  for mesh all-optical networks.
\newblock {\em Photo Network communication}, 14:111--122, 2007.

\bibitem{zhou00}
D.~Zhou and S.~Subramaniam.
\newblock Survivability in optical networks.
\newblock {\em IEEE network}, 14:16--23, Nov./Dec. 2000.

\end{thebibliography}

\end{document}